\documentclass[dvipsnames,format=sigconf]{acmart}
 
\AtBeginDocument{%
  }
\usepackage{hyperref}
\setcopyright{acmlicensed}
\copyrightyear{2025}
\acmYear{2025}
\acmDOI{XXXXXXX.XXXXXXX}

\acmConference[Conference acronym 'XX]{Make sure to enter the correct
  conference title from your rights confirmation email}{June 03--05,
  2018}{Woodstock, NY}
\acmISBN{978-1-4503-XXXX-X/18/06}

\usepackage{hyperref}
\usepackage{url}
\usepackage{graphicx}
\usepackage{algorithm}
\usepackage{algpseudocode}
\usepackage{multirow}
\usepackage{amsthm}
 
\newtheorem{Thm}{Theorem}

\def\R{{\mathbb R}}

\def\W{{\mathbb W}}

\DeclareMathOperator{\PR}{\mathbb{P}}
\DeclareMathOperator{\E}{\mathbb{E}}

\DeclareMathOperator{\RE}{\mathbb{R}}

\let\-=\mathbf
 \usepackage{pdflscape}
\makeatletter
\newcommand{\StatexIndent}[1][3]{%
	\setlength\@tempdima{\algorithmicindent}%
	\Statex\hskip\dimexpr#1\@tempdima\relax}
\algdef{S}[WHILE]{WhileNoDo}[1]{\algorithmicwhile\ #1}%
\makeatother

\begin{document}

\title{Bayesian Optimization for CVaR-based portfolio optimization}

\author{Robert Millar}
\author{Jinglai Li}
\email{rjm520@bham.ac.uk}
\email{j.li.10@bham.ac.uk}
\affiliation{%
  \institution{University of Birmingham}
  \city{Birmingham}
  \country{UK}
}

\renewcommand{\shortauthors}{Millar et al.}

\begin{abstract}
  Optimal portfolio allocation is often formulated as a constrained risk problem, where one aims to minimize a risk measure subject to some performance constraints. This paper presents new Bayesian Optimization algorithms for such constrained minimization problems, seeking to minimize the conditional value-at-risk (a computationally intensive risk measure) under a minimum expected return constraint. The proposed algorithms utilize a new acquisition function, which drives sampling towards the optimal region. Additionally, a new two-stage procedure is developed, which significantly reduces the number of evaluations of the expensive-to-evaluate objective function. The proposed algorithm's competitive performance is demonstrated through practical examples.
\end{abstract}

\keywords{Optimal portfolio allocation, Bayesian optimization}

\maketitle

\section{Introduction}
Portfolio optimization is the process of determining the optimal allocation of resources across various assets. %
A common goal is to minimize a risk measure, such as value-at-risk (VaR) or conditional value-at-risk (CVaR), while meeting a minimum expected return constraint. Traditional methods such as Linear Programming work effectively when objective and constraint functions are linear and accessible~\cite{rockafellar2000optimization,krokhmal2002portfolio}. For non-linear but accessible functions, alternate gradient descent methods have been explored (see e.g.  \cite{gaivoronski2005value,ghaoui2003worst}).

However, many practical settings involve non-linear, noisy, and expensive-to-evaluate objective and constraint functions, making traditional approaches infeasible. Bayesian Optimization (BO)~\cite{movckus1975bayesian}  has gained attention for its ability to address such challenges. BO leverages the entire history of samples to construct a posterior distribution over the unknown objective and constraint functions, employing an acquisition function to balance exploration and exploitation in selecting subsequent sampling points. 

Considerable research has focused on developing BO methods for risk-based portfolio allocation problems. Cakmak \cite{cakmak2020bayesian} proposed a BO algorithm for the unconstrained optimization of VaR and CVaR, modelling the underlying return function as a Gaussian Process and then applying a knowledge gradient-based acquisition function.

Nguyen \textit{et. al} further advanced BO methods for optimizing VaR \cite{nguyen2021value} and CVaR \cite{nguyen2021optimizing}, offering computational efficiency and theoretical robustness. More recent works include  \cite{daulton2022robust}, addressing multivariate VaR problems, and  \cite{picheny2022bayesian}, which examines scenarios where the distribution of environmental variables is unknown.
For general constrained optimization, several BO algorithms have been developed to incorporate constraints into the acquisition function design. Notable works include 
 \cite{gardner2014bayesian} and  \cite{gelbart2014bayesian}. Recent advancements  \cite{lam2017lookahead,letham2019constrained, eriksson2021scalable} have improved the efficiency of these methods.

Despite these advancements, no BO algorithm has been specifically designed for constrained portfolio optimization problems, where the constraints are integral to the problem structure and significantly impact the solution space. This work aims to fill this gap by introducing new BO methods tailored for constrained portfolio allocation. A popular class of BO methods, for general constrained optimization problems, incorporate the constraints into the acquisition function design \cite{gramacy2011optimization, gardner2014bayesian, gelbart2014bayesian}.
More recent advances include \cite{lam2017lookahead, letham2019constrained, eriksson2021scalable}, among others. 
Whilst these methods are effective, they  require frequent evaluation of the risk measure functions, which is unsuitable for complex allocation problems.

Building on \cite{gardner2014bayesian} and \cite{gelbart2014bayesian}, we aim to take advantage of two key properties which hold in portfolio allocation problems: 
1) the expected return constraint functions are much cheaper to evaluate than the objective function, i.e., the risk measures; 2) the expected return constraints are typically \emph{active} -- namely, the optimal solution lies on the boundary of the feasible region defined by the constraints.

Firstly, this paper introduces a two-stage BO adaptation, that significantly reduces computational cost by limiting full-function evaluations to samples meeting specific criteria. This differs from cascade-based BO (e.g. \citep{kusakawa2022bayesian}) where all samples in the first stage are used in the second, regardless of their feasibility or promise. Secondly, this work proposes a new acquisition function that encourages more sampling in the near-optimal region, improving the algorithm's performance. The methods are also adaptable for batch implementation to leverage parallel computing.

Numerical examples show that the proposed BO algorithms effectively solve constrained portfolio allocation problems, outperforming existing methods with a lower computational cost and faster convergence. These improvements result from combining the new acquisition function, the two-stage procedure, and parallel batch implementation.

\section{Optimal portfolio allocation}\label{sec:OPA} 
Consider an investor seeking to allocate capital across $N$ assets. Let $\textbf{x} = (x_1,...,x_N)$ be an $N$-dimensional vector representing the capital allocation or \emph{portfolio weights}, where each $x_i$ is the fraction of total capital allocated to the $i$th asset. The vector $\textbf{x}$ must satisfy the constraints $\W = \{\textbf{x} \in \R^N \mid x_i \geq 0, \sum_{i=1}^{N} x_i \leq 1\}$, ensuring that the sum of all weights does not exceed the total available capital, normalized to 1.

We account for the uncertainty in the future asset returns by introducing a random variable \textbf{Z} which follows a probability distribution $p_{\textbf{Z}}(\cdot)$. 
The return function $f(\textbf{x}, \textbf{z})$ represents the forecasted portfolio return for an allocation $\textbf{x}$ and realization $\textbf{z}$ from $\textbf{Z}$. For clarity:
\begin{itemize}
    \item $f(\textbf{0}, \textbf{z}) = 0$ indicates no capital invested means no returns.
    \item $f(\textbf{x}, \textbf{z}) < 0$ indicates a forecasted loss.
    \item $f(\textbf{x}, \textbf{z}) > 0$ indicates a forecasted gain.
\end{itemize}
For example, $f(\textbf{x}, \textbf{z}) = 0.1$ is a forecasted gain of 10\%, while $f(\textbf{x}, \textbf{z}) = -0.2$ is a forecasted loss of 20\%.

\subsection{Risk Measures}
We discuss two popular risk measures here.
Value-at-Risk (VaR) is defined as the threshold value $\omega$ such that the probability of a loss exceeding $\omega$ is at most $(1 - \alpha)$. Formally, for a return function $f$, portfolio weights $\textbf{x}$, and VaR threshold $\alpha$, VaR is defined as:\\
\centerline{
$\text{VaR}_{\alpha}[f(\textbf{x}, \textbf{Z})] = \inf \{ \omega : \mathbb{P}(f(\textbf{x}, \textbf{Z}) \leq - \omega) \leq 1 - \alpha \}$.}\\
We denote $\text{VaR}_{\alpha}[f(\textbf{x}, \textbf{Z})]$ as $v_{f}(\textbf{x};\alpha)$ for conciseness. 

Conditional Value-at-Risk (CVaR) at a specified risk level $\alpha \in (0,1)$ is the expected loss given that the losses exceed the VaR threshold. Formally, CVaR is defined as~\cite{nguyen2021optimizing},:
\\
\centerline{$\text{CVaR}_{\alpha}[f(\textbf{x}, \textbf{Z})] = -\E[f(\textbf{x}, \textbf{Z}) \mid f(\textbf{x}, \textbf{Z}) \leq -v_{f}(\textbf{x};\alpha)]$.}
\\
Before proceeding, we clarify the notation: $f(\textbf{x}, \textbf{z}) \leq 0$ indicates losses, whereas VaR and CVaR pertain to losses, so $v_{f}(\textbf{x};\alpha) \geq 0$ and $\text{CVaR}_{\alpha}[f(\textbf{x}, \textbf{Z})] \geq 0$ represent negative returns, or losses.

CVaR meets many of the desirable properties for risk measures established in \cite{artzner1999coherent}, including subadditivity, translation invariance, positive homogeneity, and monotonicity. These properties make CVaR more suitable than VaR for portfolio optimization, as VaR often exhibits multiple local extrema and unpredictable behaviour as a function of portfolio positions \cite{mausser1999beyond}. 
CVaR is usually computed with the Monte Carlo (MC) simulation, which is detailed in the Supplementary Information (SI).

\subsection{Problem Set-up}\label{Sect:PortOpt_SetUp}
The expected return for a portfolio with weights $\textbf{x}$ is defined as the expectation over all possible returns, $\E_\textbf{Z}[f(\textbf{x},\textbf{Z})]$. Research has shown a positive relationship between CVaR and expected return \cite{guo2019mean}. Increasing risk exposure generally leads to higher expected returns, and vice versa.

CVaR is monotonic to stochastic dominance of orders 1 and 2 \cite{pflug2000some} implying that if one investment option has a lower CVaR and provides equal or higher expected returns, it is universally more favourable. This property is crucial for identifying optimal portfolios that meet specific return requirements with minimal risk.

For a given expected return requirement, an optimal portfolio provides the desired return with the lowest possible CVaR. For ease of notation, we define the objective function as $g(\-x)$ and the constraint function as $R(\mathbf{x})$. The constrained portfolio optimization problem can be formulated as:
\begin{subequations}\label{e:opa}
    \begin{align}
    \underset{\textbf{x}}{\text{min }} g(\textbf{x}) & := \text{ CVaR}_{\alpha}[f(\textbf{x},\textbf{Z})] \\
    \text{s.t.} \quad R(\textbf{x}) & :=\E_\textbf{Z}[f(\textbf{x},\textbf{Z})]  \geq r^{\text{min}} \label{e:risk} \\
    \quad 0 \leq x_i \leq & 1, \, i =1,...,N, \quad \sum_{i=1}^{N} x_i \leq 1.
    \end{align}
\end{subequations}
As shown in SI, the expected return can also be computed with the MC simulation, and moreover, it is possible to obtain an accurate estimate of the expected return with a relatively low sample size, 
while a large number of samples is required to obtain an accurate estimate of the CVaR. 
\emph{As such the computational cost of calculating CVaR, i.e., the objective function, is significantly higher than the expected return constraint.}
In the numerical examples provided in Section~\ref{sec:portfolio}, the cost for evaluating the expected return is around $1\%$ of that for evaluating CVaR.
This fact is essential for our proposed BO algorithm. 

\section{Bayesian Optimization}\label{Section:BO-Original}
Bayesian Optimization \cite{movckus1975bayesian} is a powerful method for solving global optimization problems. In this section, we present an adaptation to the BO methods developed in \cite{gardner2014bayesian, gelbart2014bayesian}, so that it can handle the uncertainty caused by an environmental random variable $\textbf{Z}$.

\subsection{Gaussian Process}
A Gaussian Process (GP) is a collection of random variables, any finite number of which have a joint Gaussian distribution. The GP model provides a framework for conducting non-parametric regression in the Bayesian fashion. 
As defined in \cite{rasmussen2003gaussian}, the GP model for a function $g(\-x)$ can be written as: 
\\
\centerline{$g(\-x) \sim GP(\mu(\-x),k(\-x,\-x'))$,}
\\
where we define the mean function $\mu(\-x)$ and covariance function $k(x,x')$ for any pair of input points $\-x,\-x' \in \RE^{d}$:
\begin{equation}
	\begin{split}
		\mu(\-x) &= \E[g(\-x)],\\
		k(\-x,\-x') &= \E[(g(\-x) - \mu(\-x))(g(\-x') - \mu(\-x'))].
	\end{split}
\end{equation}
The GP-based regression proceeds as follows. 
Given a set of input points $X = \{x_1,...,x_T\}$, corresponding function values $g(\-x) = \{g(x_1),...,g(x_T)\}$,
referred to as the training set and some new design point $\hat{x}$ which we are interested in evaluating, the joint Gaussianity of all finite subsets implies:
\begin{gather}
	\begin{bmatrix} g(X) \\ g(\hat{x}) \end{bmatrix}
	= N \left(
	\begin{bmatrix}
		\mu(X) \\ \mu(\hat{x})
	\end{bmatrix},
	\begin{bmatrix}
		k(X,X) & k(X,\hat{x}) \\
		k(\hat{x},X) & k(\hat{x},\hat{x})
	\end{bmatrix}
	\right)
\end{gather}
From this, we can calculate the posterior distribution of $g(\hat{x})$  conditional on the training data set $D=\{X,\,g(X)\}$, which is also a Gaussian distribution: 
$N(\tilde{\mu}(\hat{x}), \tilde{\Sigma}(\hat{x}))$, with 
\begin{equation}
	\begin{split}
		\tilde{\mu}(\hat{x}) &= \mu(\hat{x}) + k(\hat{x},X)k(X,X)^{-1}(g(X)-\mu(X)),\\
		\tilde{\Sigma}(\hat{x}) &= k(\hat{x},\hat{x}) - k(\hat{x},X)k(X,X)^{-1}k(X,\hat{x}). 
	\end{split}
\end{equation}
Several technical issues of the GP model, such as kernel function choice and hyperparameter tuning, are not discussed here. We refer the reader to~\cite{rasmussen2003gaussian} for further details.

\subsection{Unconstrained Bayesian Optimization}
BO uses a probabilistic framework for the optimization of black-box functions, based on the GP model. In the unconstrained setting, BO sequentially evaluates the objective function at selected points, from which a GP model of the target function is constructed. The design point(s) are selected by maximizing an acquisition function, which quantifies a desired trade-off between the exploration and exploitation of the GP model. Commonly used acquisition functions include expected improvement, probability of improvement and upper confidence bounds.
The standard BO procedure for unconstrained problems is given in Alg.~\ref{alg:bo}.
\begin{algorithm}
\caption{Bayesian Optimization}\label{alg:bo}
\begin{algorithmic}
\Require objective function $g(\textbf{x})$, acquisition function $a(\textbf{x},\tilde{g})$
\Ensure a global minimizer of $g(\textbf{x})$
\State initialize the training data set $D_0$ using an initial design 
\State let $t=0$;
\While{stopping criteria not met}
   \State let $t=t+1$;
      \State construct a GP model $\tilde{g}_{t-1}$ using $D_{t-1}$;
   \State let $\textbf{x}_t = \arg\max_{\textbf{x}} a(\textbf{x},\tilde{g}_{t-1})$;
\State let $D_{t} = D_{t-1}\cup\{\textbf{x}_{t},g(\textbf{x}_t)\}$;
\EndWhile
\end{algorithmic}
\end{algorithm}

\subsection{Bayesian Optimization with Constraints}\label{sec:cbo}
In this section, we present the BO method for optimization problems with inequality constraints, largely following \cite{gardner2014bayesian} and \cite{gelbart2014bayesian}.
Suppose that we have the following constrained optimization problem:
\begin{equation}\label{e:cop}
\min_{x} ~ g(\-x) ~ \text{s.t.} ~ c_k(\-x) \leq 0, k = 1,...,K.
\end{equation}
To solve Eq.\eqref{e:cop} with the BO method, we need to model all the constraint functions $c_k(\-x)$ as GPs.
Namely, the GP model for the $k$-th constraint $c_k(\-x)$ is obtained from the constraint training set $C^k = \{(\-x_1,c_{k}(\-x_1)),...,$ $(\-x_m,c_{k}(\-x_m))\}$, where the constraint functions are evaluated at each design point. Therefore, when selecting the design points, both the objective and constraints need to be considered, which is accomplished by incorporating the constraints into the acquisition function.

The authors of \cite{gardner2014bayesian} propose modifying the \emph{Expected Improvement} (EI) acquisition function. Let ${\-x}^{+}$ be the current best-evaluated point, that is, $g(\-x^+)$ is the smallest in the current training set. We define the improvement as
\begin{equation}
I(\-x) = \max\{0,g(\-x^+)-\tilde{g}(\-x)\}
\end{equation}
where $\tilde{g}(\-x)$ is the GP model constructed with the current objective training set $D$. The EI acquisition function is defined as 
\\[0.5ex]
\centerline{
${EI}(\-x)=\mathbb{E}[I(\-x)|D]$,}
where the expectation is taken over the posterior of $\tilde{g}(\-x)$.
We further adapt this acquisition function to account for the constraints.
Let $\tilde{c}_k(\-x)$ be the GP model for the constraint function $c_k(\-x)$, conditional on the training set $C^k$, for $k=1,..., K$ and let
\\
  \centerline{  $\text{PF}(\-x) = \PR(\tilde{c}_1(\-x) \leq 0,\,\tilde{c}_2(\-x) \leq 0,\,...,\,\tilde{c}_K(\-x) \leq 0)$,}
\\
which is the probability that a candidate point $x$ satisfies all the constraints. 
In our present problem, we only need to consider the case where
the constraints are conditionally independent given $x$, as such, we have:
\begin{equation}
    \text{PF}(\-x) =\prod_{k=1}^K \PR(\tilde{c}_k(\-x) \leq 0).
\end{equation}
Finally, we define the new acquisition function to be
\begin{equation}
a_{\text{CW-EI}}(\textbf{x}) = \text{EI}(\-x)\text{PF}(\-x),\label{e:caf}
\end{equation}
which is referred to as the constraint-weighted expected improvement (CW-EI) acquisition function in \cite{gardner2014bayesian}. 
The constrained BO algorithm proceeds largely the same as the unconstrained version (Alg.~\ref{alg:bo}), except the following two main differences:
(1) the constrained acquisition function in Eq.~\eqref{e:caf} is used to select the new design points;
(2) for each design point, both the objective and constraint functions are evaluated.
We hereafter refer to this constrained BO method as CW-EI BO. 

Finally we note that in a class of BO approaches \cite{frohlich2020noisy, cakmak2020bayesian, daulton2022multi}, the underlying function $f$ is modelled as a single GP for a fixed environmental variable $\textbf{Z}$ during the optimization procedure and then $\textbf{Z}$ is only random at implementation time. 
While suitable for many unconstrained problems, this framework is inadequate for portfolio allocation problems, where $g(\textbf{x})$ and $R(\textbf{x})$ must be handled separately.

\section{Proposed methodology}\label{Sect:Proposal}
Our proposed method
builds upon CW-EI  BO presented in Section \ref{sec:cbo}.
When applied to the portfolio allocation problem, 
CW-EI BO models both the CVaR objective function and expected return constraint as separate GPs. In this approach, for each proposed weight, as determined by the acquisition function, a full evaluation of the objective and constraint functions must be performed, in order to update their respective GPs. 
Where again, we emphasise that the computational cost of calculating CVaR is significantly higher than the expected return. 
Therefore, the computational efficiency can be enhanced by reducing the number of 
CVaR evaluations.


\subsection{Activeness of the Constraint}
This section formalizes several assumptions related to the portfolio optimization problem and introduces a theorem, which enables the development of a new BO algorithm that reduces the number of CVaR evaluations.

\textbf{Assumptions}
\begin{enumerate}
    \item $f(\textbf{x},\textbf{z})$ is a continuous function of $\textbf{x}$ for any fixed $\textbf{z}$.
    \item $f(\textbf{0},\textbf{z}) \equiv 0$.
    \item For a given $\textbf{x} \in \mathbb{W}$ and any fixed $\textbf{z}$, if $f(\textbf{x},\textbf{z})\leq 0$, $f(\rho\textbf{x},\textbf{z})$ is a decreasing function of $\rho\in[0,1]$.
    \item There exists $\alpha\in(0,1)$ such that $v_f(\textbf{x};\alpha) \geq 0$ for all $\textbf{x} \in \mathbb{W}$.
\end{enumerate}

\begin{itemize}
    \item Assumption 1 ensures that small portfolio allocation changes do not lead to abrupt or unpredictable changes in outcomes, which is reasonable in most financial models.
    \item Assumption 2 is straightforward; an absence of investment results in a neutral (zero) financial return.
    \item Assumption 3 implies that if a chosen portfolio allocation results in a loss for a certain scenario, this loss does not increase if the total capital is proportionally reduced\footnote{For clarity, as $\rho$ goes from $0$ to $1$, $f$ goes from $f(0,\mathbf{z}) \equiv 0$ to $f(\mathbf{x},\mathbf{z})$. As $f(\mathbf{x},\mathbf{z})\leq 0$, the function value $f(\rho\mathbf{x},\mathbf{z})$ gets more negative, so $f$ is a decreasing function w.r.t. $\rho\in[0,1]$.}.
    \item Assumption 4 implies that there always exists a choice of $\alpha\in(0,1)$ such that, no matter the allocation $\textbf{x} \in \mathbb{W}$, $v_f(\textbf{x};\alpha)$ is positive, i.e., a loss.
\end{itemize}

From these assumptions, we derive the following theorem:

\begin{Thm}
If  function $f(\mathbf{x},\mathbf{Z})$ and distribution $p_{\-z}(\cdot)$ satisify 
assumptions 1-4, $\alpha$ is chosen such that $v_f(\mathbf{x},\alpha) \geq 0$ $\forall $ $ w \in \mathbb{W}$, and solutions to the constrained optimization problem exist, then there must exist a solution, denoted as $\mathbf{x}^*$, such that $R(\mathbf{x}^*)=r^{\text{min}}$.
\end{Thm}
    
 \begin{proof}
    First, assume that $\mathbf{x}'$ is a solution to the constrained optimization problem. It follows directly that $R(\mathbf{x}')\geq r^{\min}$. Obviously if $R(\mathbf{x}')=r^{\min}$, the theorem holds.
    
    Now consider the case that $R(\mathbf{x}')>r^{\min}$, i.e., it does not lie on the boundary of the feasible region. From assumption 1, $R(\mathbf{x})$ is a continuous function of $\mathbf{x}$ in $\mathbb{W}$. Next define a function \[h(\rho) = R(\rho\mathbf{x}')\] for $\rho \in [0,1]$. As $R(\mathbf{x})$ is a continuous function in $\mathbb{W}$, $h(\rho)$ is a continuous function too.
    
    From assumption 2, we know that $h(0) = 0$, and therefore, $$h(0)=0<r^{\min}<h(1)=R(\mathbf{x}')$$  According to the intermediate value theorem on continuous functions, there exists some $\rho^*\in(0,1)$ such that $h(\rho^*)=R(\rho^*\mathbf{x}')=r^{\min}$. Let $x^*=\rho^*\-x'$ denote this portfolio weight, which lies on the constraint boundary - we wish to compare $F(\-w^*)$ and $F(\-w')$, i.e., the CVaR values at these two portfolio weights for a fixed $\alpha$.
    
    From the Theorem's assumption, we have $v_f(\mathbf{x'},\alpha) \geq 0$ and $v_f(\mathbf{x^*},\alpha) \geq 0$. From assumption 3, we know that for any $\-z$, if $f(\-x',\-z)\leq0$, then $f(\-x',\-z)\leq f(\-x^*,\-z)\leq0$.
    
    It follows that for any $\-z\in \{\-z|f(\-x',\-z)\leq - v_f(\-x',\alpha)\}$, we have $$f(\-x',\-z)\leq f(\-x^*,\-z) \leq - v_f(\-x^*,\alpha)\leq 0.$$
      
    As such, we can derive $v_f(\-w^*,\alpha) \leq v_f(\-w',\alpha)$, and obtain, 
    \begin{align*}
    	\text{CVaR}_{\alpha}[f(\textbf{x}^*,\textbf{Z})]
        &= -\E[f(\textbf{x}^*,\textbf{Z})|f(\textbf{x}^*,\textbf{Z})\leq -v_{f}(\textbf{x}^*;\alpha)] \\
        &\leq -\E[f(\textbf{x}^*,\textbf{Z})|f(\textbf{x}^*,\textbf{Z})\leq -v_{f}(\textbf{x}';\alpha)] \\
        &\leq -\E[f(\textbf{x}',\textbf{Z})|f(\textbf{x}',\textbf{Z})\leq -v_{f}(\textbf{x}';\alpha)]  \\
        &= \text{CVaR}_{\alpha}[f(\textbf{x}',\textbf{Z})].
    \end{align*}
    Therefore, $\-x^*$ is also a minimal solution w.r.t. the objective function and $R(\-x^*)=r^{\min}$. The proof is thus complete.
    \end{proof}

Theorem 1 states that under reasonable assumptions, the constraint (Eq. ~\ref{e:risk}) is active for at least one solution. This implies that a higher expected return can only be obtained by increasing risk exposure, and thus the CVaR. The optimal solution to our problem will likely arise from an active constraint, where the minimum expected return constraint is limiting our ability to reduce the CVaR further. This provides a useful heuristic and motivates sampling towards the constraint boundary - referred to as the \emph{active region}. 

\subsection{Two-Stage Weight Selection} 
Theorem 1 suggests that we can find a solution to problem \eqref{e:opa} by searching 
along the boundary of the minimal expected return constraint. 
Specifically, based on the expected return value for a proposed portfolio weight $\textbf{x}$, we 
may decide not to evaluate the CVaR objective function in the following two situations. 
Firstly, if the expected return is lower than the minimum constraint threshold, the proposed portfolio weight is not feasible and as such, the CVaR function does not need to be evaluated. Secondly, if the expected return is too high (i.e., not approximately active), the corresponding CVaR is likely far from optimal, so the objective need not be evaluated. We account for this through the introduction of a maximum expected return parameter, denoted by $r^{\max}$, which is set on the basis that those points with expected returns higher than this parameter value are highly unlikely to be optimal for our objective. Based on these observations, we introduce a two-stage weight selection procedure. In the first stage, a portfolio weight is selected based on the acquisition function.
In the second stage, we calculate the expected return. If the expected return satisfies the requirement that \begin{equation}
r^{\min}\leq R(\textbf{x})=\E_\textbf{Z}[f(\textbf{x},\textbf{Z})]\leq r^{\max},\label{e:const_check}
\end{equation}
we complete the more expensive evaluation of our objective function, to determine the CVaR value. We then update the GP for both the constraint and objective.
If Eq.~\eqref{e:const_check} is not satisfied, we reject the proposed portfolio weight, do not evaluate the objective function and only update the GP for the expected return constraint, to ensure this weight is not re-proposed. We make no changes to the  objective function GP.
This two-stage (2S) adaptation has the advantage of only fully evaluating those points which are feasible and (approximately) active, and as such, it reduces the number of evaluations of the expensive-to-evaluate CVaR objective. The algorithm obtains two training sets, one for the CVaR objective and one for the expected return, with the former being a subset of the latter.

\subsection{New Acquisition Function}
With the two-stage selection procedure, it is clear that we will complete many more evaluations of the expected return constraint than the CVaR objective, as such, the GP for the constraint will be more accurate than that of the objective function. As a result, the CW-EI acquisition function will be highly effective at proposing feasible points, due to the quality of the constraint GP, but may be poor at proposing points with low CVaR, due to the lower quality of the objective GP. To address this, we propose a new acquisition function based on the active constraint assumption.

Namely, as the CW-EI acquisition function only accounts for the feasibility of the constraint, we want to incorporate the activeness as well. Let $\widetilde{R}(\textbf{x})$ be a GP model of the expected return $R(\textbf{x})$, we define
\begin{equation}
\text{PF}(\textbf{x}) = \PR(r^{\text{min}}\leq \widetilde{R}(\textbf{x}) \leq r^{\text{max}} )
\end{equation}
which is the probability that a design point $\textbf{x}$ is both feasible and approximately active. 
 As these two events are conditionally independent given $\textbf{x}$, we have
\begin{equation}\label{e:newpf}
    \begin{split}
    \text{PF}({\textbf{x}}) =\text{PF}_{\text{min}}({\textbf{x}})  \times \text{PF}_{\text{max}}({\textbf{x}}) \\
        \text{PF}_{\text{min}}({\textbf{x}}) = \PR(\widetilde{R}({\textbf{x}}) \geq r^{\text{min}})\\
        \text{PF}_{\text{max}}({\textbf{x}}) = \PR(\widetilde{R}({\textbf{x}}) \leq r^{\text{max}})
    \end{split}
\end{equation}
Combining Eq.~\eqref{e:newpf} with the Expected Improvement we obtain:
\begin{equation}\label{Eqs:ACQ}
	a_{\text{ACW-EI}}(\textbf{x}) = \text{EI}(\textbf{x})\text{PF}_{\text{min}}(\textbf{x})\text{PF}_{\text{max}}(\textbf{x}),
\end{equation}
which is hereafter referred to as the \emph{active constraint-weighted expected improvement} (ACW-EI) acquisition function.
Note that this acquisition function depends on both the GP models for CVaR and the expected return. In this paper, we write it as $a_{\text{ACW-EI}}(\textbf{x},\widetilde{g},\widetilde{R})$.

The new term $\text{PF}_{\max}$ in the acquisition function encourages the proposed points to be approximately active, which, by proxy, increases the likelihood that such a point is near-optimal with respect to the risk measure objective function. The choice of $r^\text{max}$ is explored through additional numerical examples in the SI. The inclusion of this parameter is a crucial aspect of our proposed BO algorithms. Two feasible points with different true objective function values are likely to have similar expected improvement values (prior to full evaluation), due to the low-quality GP for the objective function and equal probability of feasibility for the constraint. As such, in the existing methodology, the two points may be considered equally. By introducing the new $r^{\text{max}}$ term - based on the more accurate expected return GP - our proposed BO procedure is able to differentiate between these two points during the selection procedure.

\subsection{The complete algorithm}
To complete our proposed algorithm, we must discuss the summation constraint:\\[0.5ex]
\centerline{$0\leq x_i\leq 1,\, i=1,...,N,\, \sum_{i=1}^N x_i \leq 1,$}\\[0.5ex]
which will be denoted as $\textbf{x}\in S$ in what follows. 
It is possible to deal with these constraints in the same manner as the expected return, i.e., as GP models. However, unlike the expected return constraint, which is probabilistic, the summation constraint is deterministic and easy to evaluate. As such, we impose the constraint during the maximisation of the acquisition function, by solving the following constrained maximization problem:
$\max_{\textbf{x}\in S} a_{\text{ACW-EI}} (\textbf{x})$,
which in this work is solved with the barrier method. 

Finally, by combining the two-stage point selection, the ACW-EI acquisition function, 
and the constrained acquisition maximization, we obtain a complete 2S-ACW-EI BO algorithm, detailed in Alg.~\ref{alg:2s-acw-ei}.


\begin{algorithm}
\caption{The 2S-ACW-EI BO algorithm}\label{alg:2s-acw-ei}
\begin{algorithmic}
\State Initialize the training data sets $D$ (for the objective) and $C$ (for the constraint),
 using an initial design;
\State Let $t=1$;
\While{stopping criteria not met}
      \State Construct an objective function GP model $\widetilde{g}_{t-1}$ using $D$;
      \State Construct a constraint GP model $\widetilde{R}_{t-1}$ using $C$;
   \State Let $\hat{\textbf{x}} = \arg\max_{\textbf{x}\in S} a_{\text{ACW-EI}}(\textbf{x},\widetilde{g}_{t-1},\widetilde{R}_{t-1})$;
\State Evaluate the constraint $R(\hat{\textbf{x}})$;
\State Let {$C = C\cup\{\hat{\textbf{x}},R(\hat{\textbf{x}})\}$;}
\If{$r^{\min}\leq R(\hat{\textbf{x}})\leq r^{\max}$}
\State Evaluate the objective $g(\hat{\textbf{x}})$;
\State Let $D = D\cup\{\hat{\textbf{x}},g(\hat{\textbf{x}})\}$;
\State let $t=t+1$;
\EndIf
\EndWhile
\end{algorithmic}
\end{algorithm}

\subsection{Batch Implementation}
In most BO approaches, one uses an acquisition function to select a single point to evaluate. From which, the posterior GPs are updated and the process is repeated. This is \emph{sequential}, as each point is selected and evaluated one at a time.

It is expensive to evaluate the objective function, and as such, it may be advantageous to evaluate several points simultaneously, for example using parallel computers. 
In this regard, a batch implementation of BO is desirable, where several design points are selected using the acquisition function and then evaluated simultaneously in parallel. In this section, we discuss a batch implementation for our proposed algorithms.

In most batch BO methods, the batch of design points is determined sequentially via a given point-selection procedure, from which the objective and constraint functions are evaluated after the whole batch is obtained. 
In our two-stage method, evaluation of the expected return constraint is included in the point-selection procedure and once the whole batch is obtained, the CVaR objective is evaluated in parallel.
More specifically, the expected return is evaluated for each new proposed point. If the expected return satisfies Eq.~\eqref{e:const_check}, it is added to the batch and the constraint GP is updated. 
If the expected return does not satisfy Eq.~\eqref{e:const_check}, the point is not added to our batch but the GP for the constraint is updated, to ensure that the point is not proposed again. 
Once a batch has been determined, each point is fully evaluated - knowing that all batch points are both feasible and approximately active. The pseudo-code for our two-stage batch selection is provided in Alg.~\ref{alg:batch}.

As the batch approach can be implemented in parallel, it has a lower computational cost. However, the batch approach requires a greater total number of samples to converge to the optimal solution - as demonstrated in our numerical examples - due to the GPs being updated less frequently, so each sample is chosen on the basis of a less accurate GP compared to at the equivalent stage in the sequential approach.

\kern -1\medskipamount
\begin{algorithm}
\caption{Two-Stage Batch Selection}\label{alg:batch}
\begin{algorithmic}
\Require  
a training set for the CVaR objective function $D$, 
a training set for the expected return constraint $C$
\Ensure a batch of $b$ design points,
\State let $B=\emptyset$
\State let $i=0$;
\While{$i<b$}
      \State propose a new design point $\hat{\textbf{x}}$ based on  a prescribed selection rule;
   \State evaluate the constraint $R(\hat{\textbf{x}})$;
   \If{$r^{\min}\leq R(\hat{\textbf{x}}) \leq r^{\max}$}
\State let $B = B\cup\{\hat{\textbf{x}}\}$;
   \State let $i=i+1$;
 \EndIf
 \State let $C=C\cup\{\hat{\textbf{x}},R(\hat{\textbf{x}})\}$;
 \State update the GP model for the constraint using $C$;
\EndWhile
\end{algorithmic}
\end{algorithm}

\kern-\medskipamount
\section{Numerical Experiments}
In this section, we apply the proposed BO algorithms to several numerical examples. BO was implemented using \emph{Trieste} \cite{Berkeley_Trieste_2023}, a BO Python package built on TensorFlow. We used the default Matern 52 Kernel with a length scale of 1.0 and noise variance of $10^{-7}$. For acquisition maximization, the summation constraint was included as a barrier function, and the problem was solved using the Efficient Global Optimization method provided by the package.

\kern-\medskipamount
\subsection{Mathematical example}
We first consider a simple mathematical example, to demonstrate how the design points are selected by the different methods. Adapted from \cite{gramacy2016modeling}, we seek to solve the following constrained optimization problem:
\begin{equation}
    \begin{split}
        \underset{\textbf{x}}{\min} \; f(\textbf{x}):= & -x_1 - x_2 \\
        \text{s.t.} ~ c(\textbf{x}) := & \frac{3}{2} - x_1 - 2x_2 
         - \frac{1}{2} \text{sin}(2\pi (x_1^{2}-2x_2))\geq 0
    \end{split}
    \label{eq:ME}
\end{equation}
The solution to the problem is $\textbf{x} = (0.918, 0.540)$, where $f(\textbf{x}) = 1.458$. The original CW-EI method, ACW-EI (i.e. the new acquisition function without the 2S process), and 2S-ACW-EI each use 10 initial points and then a further 50 iterations. Figure \ref{fig:Gramacy} shows the design points obtained by each of the three algorithms.

\begin{figure*}
 \centerline{ \includegraphics[width=.9\textwidth]{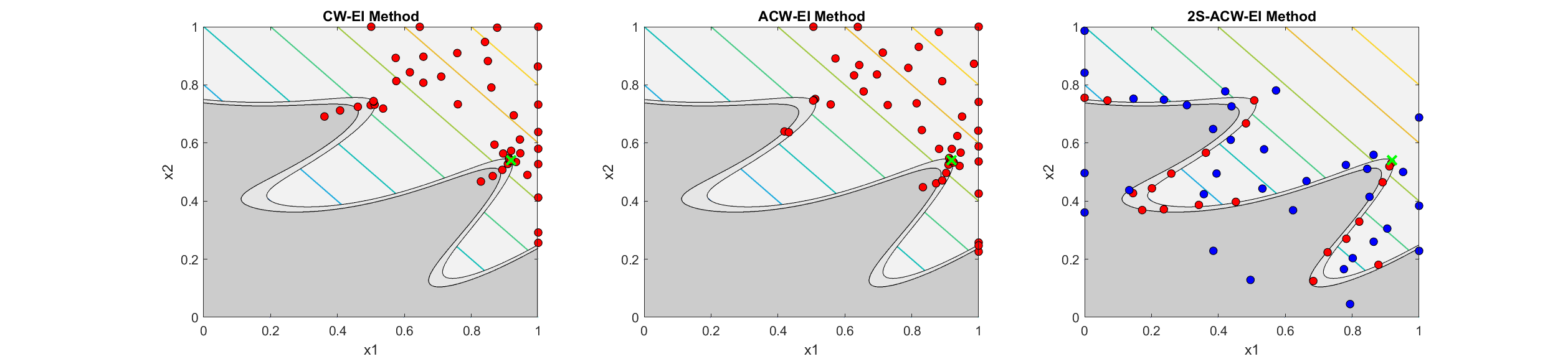}}
  \caption{Plots showing the optimal solution (green-x) for numerical example one and the design points generated by each of the three methods. The figures include both the fully evaluated points (red) and those for which only the constraint was evaluated (blue). The feasible region is dark grey, the active region is light grey and the infeasible region is white. The objective function contours are shown too.}
  \label{fig:Gramacy}
\end{figure*}

The CW-EI and ACW-EIs methods perform similarly in this task, where the algorithms generate a significant number of infeasible samples with high objective value, before moving towards the feasible region. Both methods first establish a good GP for the objective function, encouraging samples to be generated in the high objective region, before the GP for the constraint is fully formed.  
In contrast, in the 2S-ACW-EI method, samples are only fully evaluated if they are in the active region, therefore after a few iterations, the GP for the objective and constraint functions are weak and strong respectively. Thanks to the well-formed GP model for the constraint, the acquisition function prioritises the generation of points in the feasible region, in particular, in the active region, before finding those feasible points which are maximised for the objective.

\subsection{Portfolio allocation examples}\label{sec:portfolio} 

\subsubsection{Problem setup}
The following three examples are based on an investor seeking to optimally allocate capital to stock or stock options, related to the twenty largest technology companies listed on American stock exchanges (both the NYSE \& Nasdaq) by market capitalisation. We take $\textbf{Z}$ to be the stock price at the future time, assumed to be normally distributed, where the distribution parameters are determined by historical data (see Table 1 in SI).

In all three  examples, the return function is
\begin{equation}
    f(\textbf{x},\textbf{z}) = \sum_{i=1}^{20} x_{i}y_{i}(z_i),
\end{equation}
where $y_i$ is the asset return - stated as a ratio, rather than absolute value - corresponding to the $i$-th company, a function of its future stock price $z_i$.
In the three examples, we alter the asset type -- namely the function $y_i(z_i)$ varies. In each example, we consider a lower and higher return constraint. 

\textbf{Example One.} The investor's capital is allocated directly to the twenty stocks, with $y_i(z_i):=z_i/\bar{z}_i$, where $\bar{z}_i$ is the stock's purchase price. The constraints are set for $1 - \alpha = 0.0001$ with $r^{min} =$ (a) $1.45$ and (b) $1.55$.

\textbf{Example Two.} The capital is allocated to European Call options based on the twenty stocks, held until expiry. A European Call option gives the owner the right to purchase the underlying asset at a pre-agreed strike price on a specified future date. If the current bid price of the call option for the $i$-th stock is $b_i$ and the strike price is $K_i$, the asset return is:\\
\centerline{   $ y_i(z_i) := \frac{\max(0, z_i - K_i) - b_i}{b_i}$.}\\
The constraints are set for $1 - \alpha = 0.0001$ with $r^{min} =$ (a) $5.30$ and (b) $5.40$.

\textbf{Example Three.} The return is derived from selling European Call options after six months rather than holding them to maturity. The return depends on the change in the option price, modelled using quadratic functions of the underlying asset returns, realized through a delta-gamma approximation. The associated call option return becomes:
\\
  \centerline{ $y_i(z_i) := \Delta_i\;\epsilon + \frac{1}{2}\;\Gamma_i\;\epsilon^{2}$,}\\
where $\epsilon=z_i-\bar{z}_i$. The constraints are set for $1 - \alpha = 0.0001$ with $r^{min} =$ (a) $2.90$ and (b) $3.00$.
\begin{figure*}
 \centerline{ \includegraphics[width=0.73\textwidth]{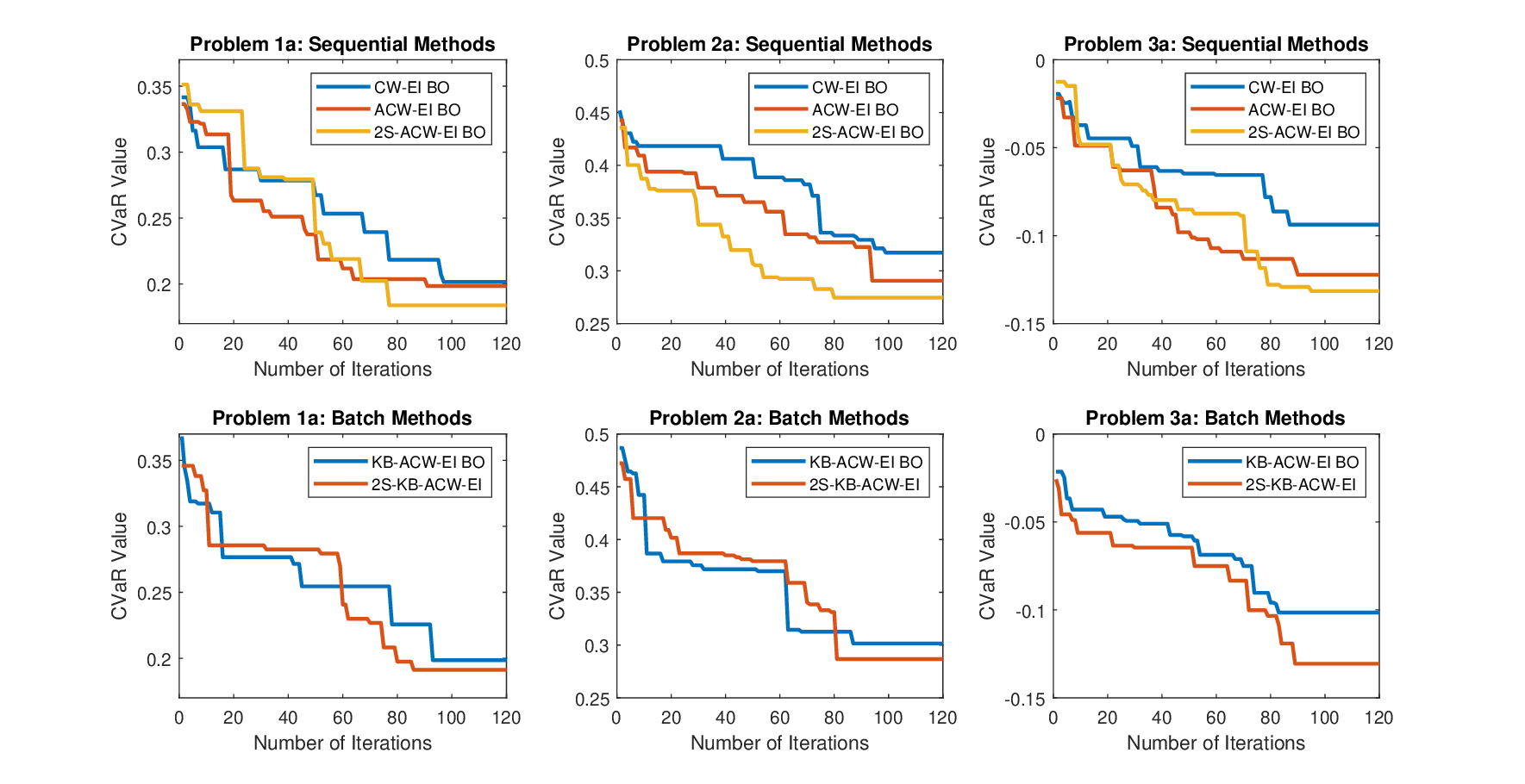}}
 \centerline{ \includegraphics[width=0.73\textwidth]{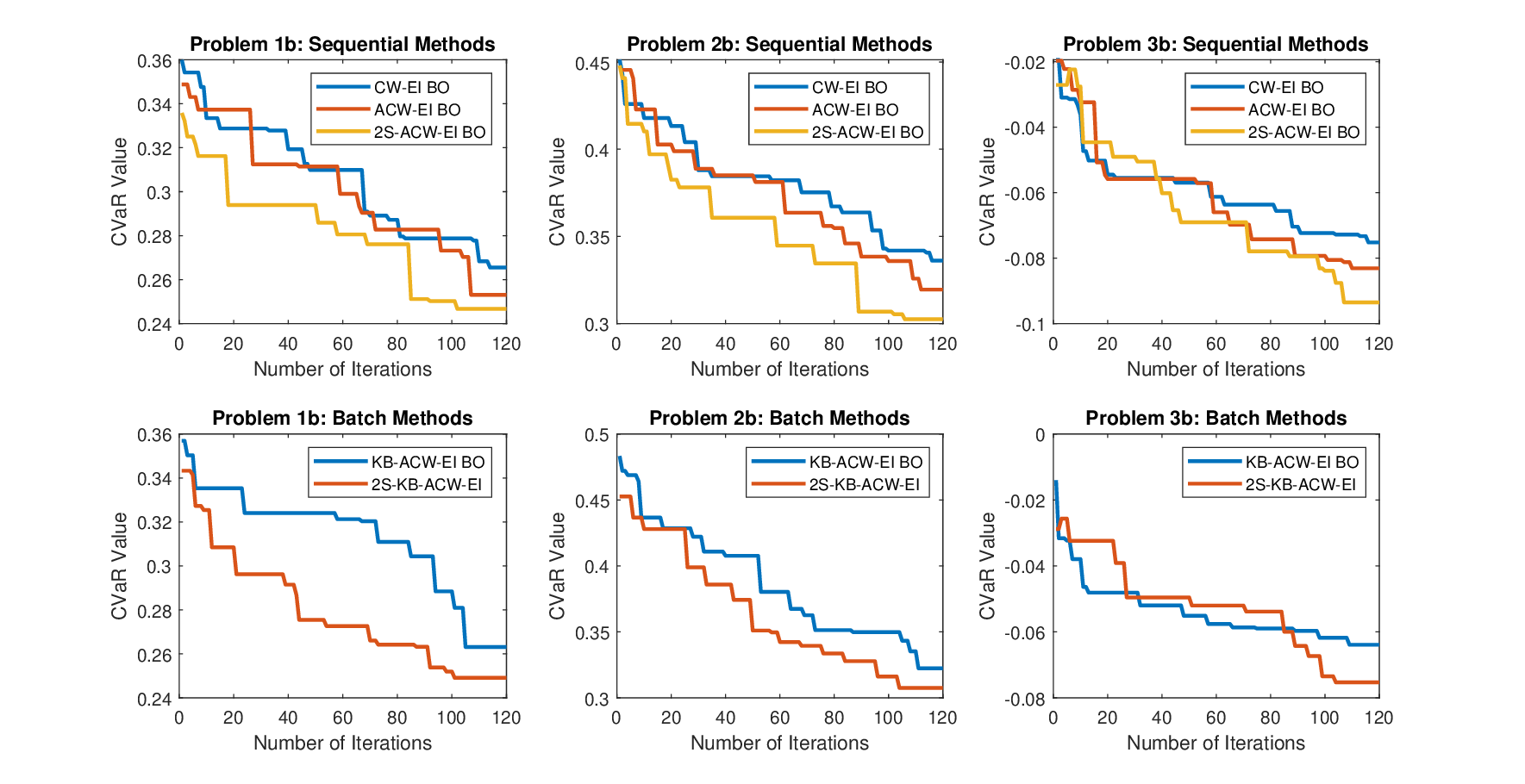}}
 
  \kern -4\medskipamount
  \caption{The best objective value obtained after each iteration for the portfolio allocation problems across the existing method (CW-EI BO) and the four new proposed methods.}
  \label{fig:Best}
\end{figure*}

\subsubsection{Experimental Results}
In all three examples, we applied the three sequential and two batch methods. We used 10 initial portfolio weights, 110 iterations for the sequential methods, and 11 batches of size 10 for the batch methods. We set $r^{\max} = 110\% r^{\min}$ in our numerical experiments. All experiments were repeated 20 times. The results are given in Table~\ref{TAB:Result_Port_All}.

\begin{table*}[]
\centering
\resizebox{0.69\textwidth}{!}
{
\begin{tabular}{|l|lll|ll|}
\hline
\multicolumn{1}{|l|}{} & \multicolumn{3}{l|}{\textbf{Sequential BO Methods}} & \multicolumn{2}{l|}{\textbf{Batch BO Methods}} \\
\cline{2-6}
\multicolumn{1}{|l|}{} & \multicolumn{1}{l|}{\textbf{CW-EI}} & \multicolumn{1}{l|}{\textbf{ACW-EI}} & \textbf{2S-ACW-EI} & \multicolumn{1}{l|}{\textbf{KB-ACW-EI}} & \textbf{2S-KB-ACW-EI} \\
\hline\hline
\multicolumn{1}{|l|}{\textbf{1a CVaR (SD)}} & \multicolumn{1}{l|}{$0.202~(0.013)$} & \multicolumn{1}{l|}{$0.199~(0.013)$} & $\mathbf{0.184}~(0.012)$ & \multicolumn{1}{l|}{$0.199~(0.012)$} & $0.191~(0.012)$ \\ 
\hline
\multicolumn{1}{|l|}{\textbf{1a Ex Return (SD)}} & \multicolumn{1}{l|}{$1.473~(0.012)$} & \multicolumn{1}{l|}{$1.485~(0.012)$} & $1.473~(0.012)$ & \multicolumn{1}{l|}{$1.479~(0.012)$} & $1.478~(0.012)$ \\
\hline\hline
\multicolumn{1}{|l|}{\textbf{1b CVaR (SD)}} & \multicolumn{1}{l|}{$0.266~(0.012)$} & \multicolumn{1}{l|}{$0.253~(0.012)$} & $\mathbf{0.247}~(0.012)$ & \multicolumn{1}{l|}{$0.263~(0.014)$} & $0.249~(0.013)$ \\
\hline
\multicolumn{1}{|l|}{\textbf{1b Ex Return (SD)}} & \multicolumn{1}{l|}{$1.581~(0.012)$} & \multicolumn{1}{l|}{$1.577~(0.012)$} & $1.561~(0.012)$ & \multicolumn{1}{l|}{$1.580~(0.012)$} & $1.567~(0.012)$ \\
\hline\hline
\multicolumn{1}{|l|}{\textbf{2a CVaR (SD)}} & \multicolumn{1}{l|}{$0.317~(0.013)$} & \multicolumn{1}{l|}{$0.291~(0.015)$} & $\mathbf{0.275}~(0.014)$ & \multicolumn{1}{l|}{$0.302~(0.013)$} & $0.287~(0.013)$ \\
\hline
\multicolumn{1}{|l|}{\textbf{2a  Ex Return (SD)}} & \multicolumn{1}{l|}{$5.335~(0.013)$} & \multicolumn{1}{l|}{$5.320~(0.012)$} & $5.302~(0.013)$ & \multicolumn{1}{l|}{$5.341~(0.012)$} & $5.322~(0.012)$ \\
\hline\hline
\multicolumn{1}{|l|}{\textbf{2b CVaR (SD)}} & \multicolumn{1}{l|}{$0.336~(0.014)$} & \multicolumn{1}{l|}{$0.320~(0.014)$} & $\mathbf{0.303}~(0.013)$ & \multicolumn{1}{l|}{$0.322~(0.013)$} & $0.308~(0.013)$ \\
\hline
\multicolumn{1}{|l|}{\textbf{2b Ex Return (SD)}} & \multicolumn{1}{l|}{$5.427~(0.013)$} & \multicolumn{1}{l|}{$5.428~(0.012)$} & $5.417~(0.013)$ & \multicolumn{1}{l|}{$5.433~(0.013)$} & $5.420~(0.012)$ \\
\hline\hline
\multicolumn{1}{|l|}{\textbf{3a CVaR (SD)}} & \multicolumn{1}{l|}{$-0.094~(0.012)$} & \multicolumn{1}{l|}{$-0.122~(0.014)$} & $\mathbf{-0.132}~(0.013)$ & \multicolumn{1}{l|}{$-0.102~(0.012)$} & $-0.131~(0.014)$ \\
\hline
\multicolumn{1}{|l|}{\textbf{3a  Ex Return (SD)}} & \multicolumn{1}{l|}{$3.105~(0.013)$} & \multicolumn{1}{l|}{$3.030~(0.013)$} & $2.938~(0.012)$ & \multicolumn{1}{l|}{$3.082~(0.013)$} & $2.97~(0.013)$ \\
\hline\hline
\multicolumn{1}{|l|}{\textbf{3b CVaR (SD)}} & \multicolumn{1}{l|}{$-0.075~(0.013)$} & \multicolumn{1}{l|}{$-0.083~(0.013)$} & $\mathbf{-0.094}~(0.014)$ & \multicolumn{1}{l|}{$-0.064~(0.012)$} & $-0.075~(0.013)$ \\
\hline
\multicolumn{1}{|l|}{\textbf{3b  Ex Return (SD)}} & \multicolumn{1}{l|}{$3.113~(0.013)$} & \multicolumn{1}{l|}{$3.075~(0.013)$} & $3.056~(0.012)$ & \multicolumn{1}{l|}{$3.125~(0.012)$} & $3.089~(0.013)$ \\ 
\hline
\end{tabular}
}
\caption{Average of the best objective and constraint values across repeated experiments: in each case, the best result among the methods is shown in bold. The standard deviations are given in parentheses.}
\label{TAB:Result_Port_All}
\end{table*}

For all three examples, our proposed sequential methods outperformed the standard BO approach, finding a lower CVaR objective value while meeting the feasibility condition. Additionally, the two-stage approach produced better results than the one-stage approach. The same trend was observed for the batch methods, with the two-stage method outperforming the one-stage method. The batch methods obtained better results than the standard sequential BO method but performed worse than the best sequential implementations. This outcome is as expected, due to the GP only being updated after a full batch of samples is identified, whereas in the sequential approach, the GP is updated for each new sample. Using parallel implementation, the batch method is significantly faster than the sequential approach. To further illustrate the results, we plot the best solution's objective value after each iteration in Fig.~\ref{fig:Best}. Consistently, the best solution of 2S-ACW-EI decreases faster than the other two sequential methods. The two-stage batch method performs better than the standard implementation in all cases.

\begin{figure*}
 \centerline{ \includegraphics[width=0.73\textwidth]{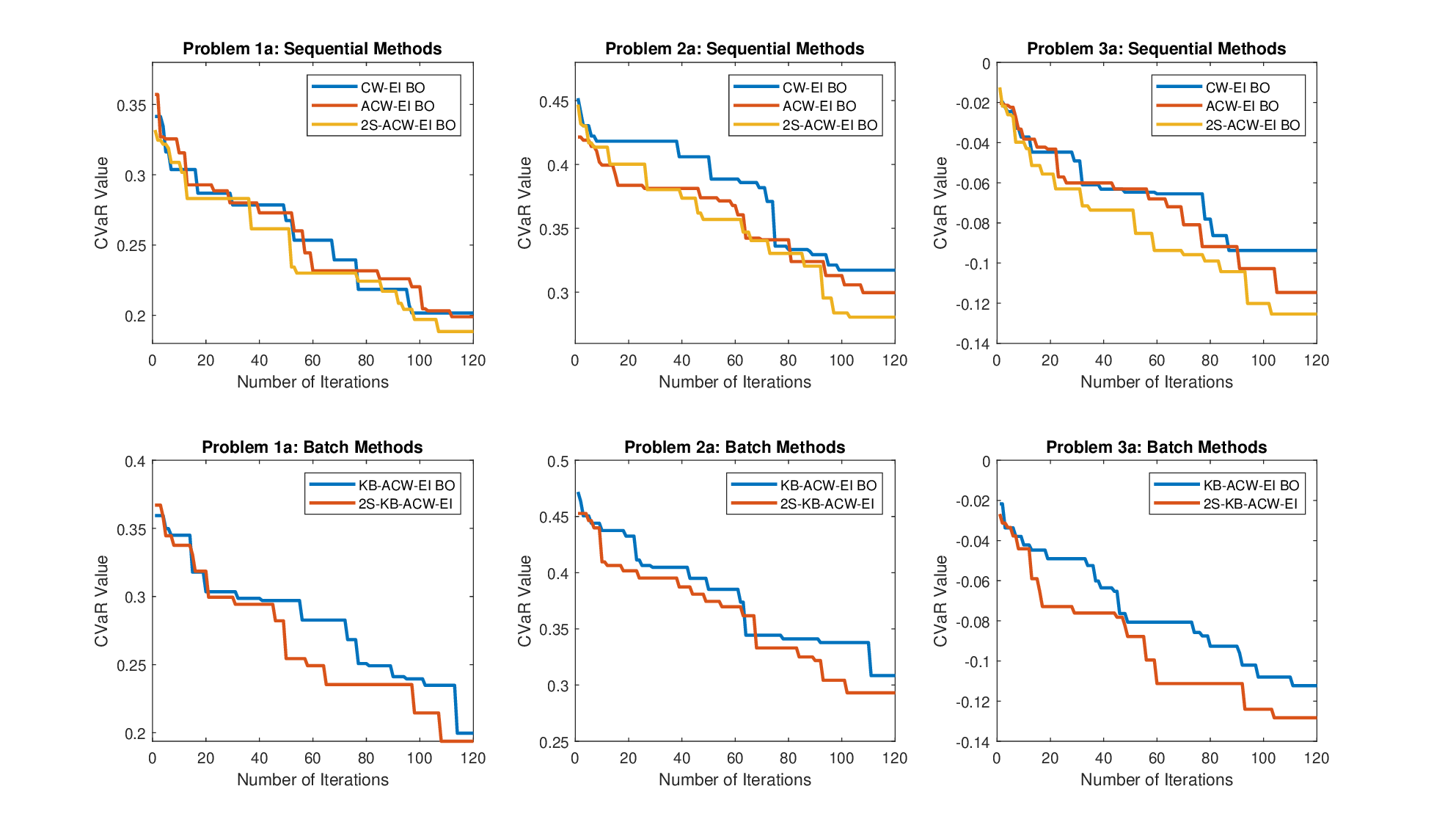}}

 \kern-4\medskipamount
  \caption{The best objective value obtained after each iteration for problems 1a - 3a, with $r^{\max} = 105\% r^{\min}$.}
  \label{fig:more}
\end{figure*}
    \begin{table*}[]
    \centering
    \resizebox{.69\textwidth}{!}
        {
        \begin{tabular}{|l|lll|ll|}
        \hline
        \multicolumn{1}{|l|}{} & \multicolumn{3}{l|}{\textbf{Sequential BO Methods}} & \multicolumn{2}{l|}{\textbf{Batch BO Methods}} \\
        \cline{2-6}
        \multicolumn{1}{|l|}{} & \multicolumn{1}{l|}{\textbf{CW-EI}} & \multicolumn{1}{l|}{\textbf{ACW-EI}} & \textbf{2S-ACW-EI} & \multicolumn{1}{l|}{\textbf{KB-ACW-EI}} & \textbf{2S-KB-ACW-EI} \\
        \hline\hline
        \multicolumn{1}{|l|}{\textbf{1a CVaR (SD)}} & \multicolumn{1}{l|}{$0.202~(0.013)$} & \multicolumn{1}{l|}{$0.198~(0.011)$} & $\mathbf{0.188}~(0.014)$ & \multicolumn{1}{l|}{$0.201~(0.0.014)$} & $0.194~(0.011)$ \\ 
        \hline
        \multicolumn{1}{|l|}{\textbf{1a Ex Return (SD)}} & \multicolumn{1}{l|}{$1.473~(0.012)$} & \multicolumn{1}{l|}{$1.473~(0.018)$} & $1.471~(0.013)$ & \multicolumn{1}{l|}{$1.477~(0.016)$} & $1.474~(0.017)$ \\
        \hline\hline
        \multicolumn{1}{|l|}{\textbf{2a CVaR (SD)}} & \multicolumn{1}{l|}{$0.317~(0.013)$} & \multicolumn{1}{l|}{$0.299~(0.014)$} & $\mathbf{0.281}~(0.014)$ & \multicolumn{1}{l|}{$0.308~(0.011)$} & $0.293~(0.017)$ \\
        \hline
        \multicolumn{1}{|l|}{\textbf{2a  Ex Return (SD)}} & \multicolumn{1}{l|}{$5.335~(0.013)$} & \multicolumn{1}{l|}{$5.324~(0.013)$} & $5.317~(0.016)$ & \multicolumn{1}{l|}{$5.331~(0.012)$} & $5.323~(0.013)$ \\
        \hline\hline
        \multicolumn{1}{|l|}{\textbf{3a CVaR (SD)}} & \multicolumn{1}{l|}{$-0.094~(0.012)$} & \multicolumn{1}{l|}{$-0.115~(0.016)$} & $\mathbf{-0.125}~(0.013)$ & \multicolumn{1}{l|}{$-0.112~(0.014)$} & $-0.128~(0.015)$ \\
        \hline
        \multicolumn{1}{|l|}{\textbf{3a  Ex Return (SD)}} & \multicolumn{1}{l|}{$3.105~(0.013)$} & \multicolumn{1}{l|}{$3.103~(0.014)$} & $3.083~(0.018)$ & \multicolumn{1}{l|}{$3.102~(0.018)$} & $3.061~(0.013)$ \\
        \hline
        \end{tabular}
        }
        \caption{Same results as those in Table \ref{TAB:Result_Port_All}, for problems 1a-3a, obtained with $r^{\max} = 105\% r^{\min}$.}
    \kern -6\medskipamount
        
    \label{TAB:Result_Port_All_2}
    \end{table*}

Finally, regarding the choice of $r^{\max}$; our numerical experiments found that setting $r^{\max} = 110\% r^{\min}$ generally works well. To test more rigorously how sensitive our proposed BO algorithm is to this parameter, we provide further numerical results obtained with $r^{\max} = 105\% r^{\min}$, in Table~\ref{TAB:Result_Port_All_2}, as well as graphical results in Fig.~\ref{fig:more}. 
These results are quantitatively similar to those with $r^{\max} = 110\% r^{\min}$, showing that the proposed algorithms are not highly sensitive to the choice of $r^{\max}$.

\section{Conclusion} \label{Section:Conclusion}
In this paper, we addressed the optimal portfolio allocation problem, which aims to minimize a computationally intensive risk measure under a minimum expected return constraint. We proposed four new BO algorithms specifically designed for such problems, significantly reducing the number of evaluations of the expensive objective function. These methods leverage the special properties of portfolio optimization problems by developing a new acquisition function, a two-stage portfolio weight selection process, and a batch implementation that takes advantage of parallel computing.


Several future directions can enhance our proposed methods. Firstly, the proposed method may not find the optimal solution for problems where the solutions do not lie on the boundary of the expected return constraint. In critical applications such as automated investing, returning a sub-optimal solution can have serious consequences. Therefore, it is crucial to implement mechanisms to ensure the reliability and safety of the algorithms. A heuristic strategy is to search around the obtained solution to find a better one. This issue should be carefully studied in the future.


Overall, our proposed BO algorithms offer a promising approach to solving the computationally intensive risk minimization problems in portfolio optimization, with potential applications across various fields.
\bibliographystyle{ACM-Reference-Format}
\bibliography{Bib_BO_MMC}
\newpage
\appendix

\bigskip

\centerline{\bf\LARGE Supplementary Information}

\subsection*{A.1 Monte Carlo Estimation of CVaR}\label{Section:MC}
Within Bayesian Optimization (BO), the acquisition function is used to propose a new design point $\textbf{x}$, from which the complete return distribution $f(\textbf{x},\textbf{Z})$ can be obtained. A Monte Carlo (MC) simulation can be used to construct this distribution, from which the expected return and CVaR can be obtained - the details of which we provide here.

Let $\mathbf{x}$ be the proposed design point. The return function is subject to uncertainty from the environmental variable $\textbf{Z}$ which follows a probability distribution $p_\textbf{z}(\cdot)$. Let $y$ be a scalar characterised by the return function $y = f(\textbf{x},\textbf{Z})$. We want to determine the probability density function (PDF) of y, given by $\pi(y)$, where $\mathbf{z}$ and $y$ are continuous random variables.

For the sake of convenience, we assume that $\pi(y)$ has a bounded support $R_y=[a,b]$. If the support of $\pi(y)$ is not bounded, we choose an interval $[a,b]$ that is sufficiently large so that $\mathbb{P}(y\in[a,b])\approx1$.
We first decompose $R_y$ into $M$ bins of equal width $\Delta$ centred at the discrete values $\{b_1,...,b_M\}$ and define the $i$-th bin as the interval $B_i = [b_i - \Delta/2, b_i + \Delta/2]$.
This binning implicitly defines a partition of the input space into $M$ domains $\{D_i\}_{i=1}^{M}$, where
    \begin{equation}
		D_i = \{\mathbf{z} \in \textbf{Z} : f(\textbf{x},\textbf{z}) \in B_i\}
	\end{equation}
is the domain in $\textbf{Z}$ that maps into the $i$-th bin $B_i$ (see Fig.~\ref{fig:MultiBin}).

\begin{figure}[h]
	\centering
	\includegraphics[width=0.8\columnwidth]{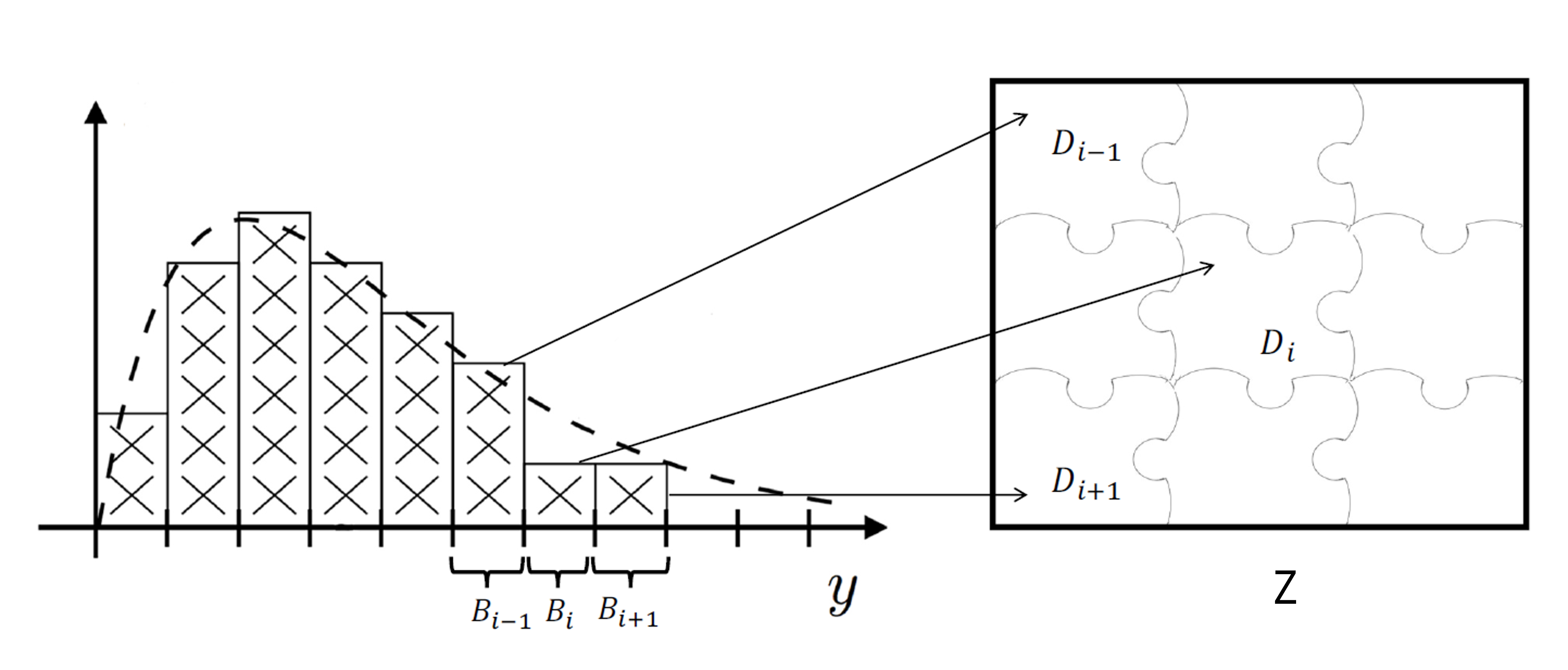}
	\caption{Schematic illustration of the connection between $B_i$ and $D_i$. }
	\label{fig:MultiBin}		
\end{figure}

While $B_i$ are simple intervals, the domains $D_i$ are multidimensional regions with possibly tortuous topologies. Therefore, an indicator function is used to classify whether a given $\mathbf{z}$-value is in bin $D_i$ or not. Formally, the indicator function is defined as
	\begin{equation}
		I_{D_i}(\mathbf{z}) = \begin{cases}
			1, & \text{if $\mathbf{z} \in D_i$};\\
			0, & \text{otherwise}
		\end{cases}
	\end{equation}
or equivalently $\{y = f(\textbf{x},\textbf{z}) \in B_i\}$.
By using this indicator function, the probability  that $y$ is in the $i$-th bin, i.e. $P_i = \PR\{y \in B_i\}$, can be written as an integral in the input space:
	\begin{equation}
		P_i = \int_{D_i}p(\mathbf{z})d\textbf{z} = \int I_{D_i}(\mathbf{z})p(\mathbf{z})d\textbf{z} = \E[I_{D_i}(\mathbf{z})].
		\label{Pi_estimator}
	\end{equation}
We can estimate $P_i$ via a standard MC simulation. 
Namely, we draw $N$ i.i.d. samples $\{\mathbf{z}_1,...,\mathbf{z}_N\}$ from the distribution $p(\mathbf{z})$,
and calculate the MC estimator of $P_i$ as 
\begin{equation}
		\hat{P}_i^{MC} = \frac{1}{N} \sum_{j=1}^{N} I_{D_i}(\mathbf{z}^j) = \frac{N_i}{N},\quad \mathrm{for}\,\, i=1,\,...,M,
		\label{MC_estimator}
	\end{equation}
where $N_i$ is the number of samples that fall in bin $B_i$.

Once we have obtained  $\{P_i\}_{i=1}^M$, the PDF of $y$ at the point $y_i\in B_i$ - for a sufficiently small $\Delta$ - can be calculated as $\pi(y_i) \approx P_i / \Delta$. The expected return and conditional value-at-risk are then obtained from the distribution $\pi(y_i)$.

Crucially, to just obtain the expected return, for a proposed weight \textbf{x}, one only needs a relatively small number of samples to obtain an accurate estimate through a Monte Carlo simulation. In contrast, to determine an accurate estimation of the CVaR, one must use a large number of samples from $\textbf{Z}$, to ensure that the tail distribution is properly assessed.

\subsection*{A.2 Portfolio Allocation Data Details}
We provide further detail of the portfolio allocation examples included within the main paper. Our portfolio allocation example is based on the twenty largest technology companies listed on American stock exchanges (both the NYSE \& Nasdaq) by market capitalisation, as is shown in Table~\ref{tab:FinData}. Data is taken from July 13, 2022.

We assume that the future stock price follows a normal distribution whose mean and standard deviation are stated in Table~\ref{tab:FinData}, based on historical data. We further assume that all stock prices are independent of each other. Further parameter values, such as the bid price and the strike price, are also shown in the table. The stock price data, including historic performance, is from \emph{morningstar.com}; call options price data is from \emph{marketwatch.com} and Greeks' data is from \emph{nasdaq.com}.

\begin{landscape}
\begin{table}
\resizebox{1.5\textwidth}{!}
{ 
\begin{tabular}{|l|l|l|l|l|l|l|l|l|l|}
\hline
\textbf{Asset $i$} &
\textbf{Company Name} &
  \textbf{Ticker} &
  \textbf{Stock Price (\$)} &
  \textbf{\begin{tabular}[c]{c} Historic Average \\ Annual Return (\%)\end{tabular}} &
  \textbf{\begin{tabular}[c]{c} Historic Return \\ Std Dev (\%)\end{tabular}} &
  \textbf{Strike Price (\$)} &
  \textbf{\begin{tabular}[c]{c}12-month Call \\ Bid Price (\$)\end{tabular}} &
  \textbf{Delta} &
  \textbf{Gamma} 
  \\ 
\hline
$1$ & Apple Inc             & AAPL  & $145.49$  & $34.67$  & $66.63$  & $160$   & $14.60$  & $0.4462$ & $0.0112$ \\ \hline
$2$ & Microsoft Corp        & MSFT  & $252.72$  & $31.83$  & $42.45$  & $275$   & $21.45$  & $0.4029$ & $0.0068$ \\ \hline
$3$ & Alphabet Inc          & GOOGL & $2227.07$ & $29.07$  & $40.46$  & $2,450$ & $222.20$ & $0.4249$ & $0.0007$ \\ \hline
$4$ & Amazon.com Inc        & AMZN  & $110.40$  & $75.21$  & $196.12$ & $120$   & $15.10$  & $0.4615$ & $0.0119$ \\ \hline
$5$ & Tesla Inc             & TSLA  & $711.12$  & $116.93$ & $219.27$ & $780$   & $152.90$ & $0.5313$ & $0.0012$ \\ \hline
$6$ & Meta Platforms Inc    & META  & $163.49$  & $36.96$  & $33.72$  & $180$   & $26.05$  & $0.5160$ & $0.0066$ \\ \hline
$7$ & Nvidia Corp           & NVDA  & $151.64$  & $59.05$  & $89.51$  & $165$   & $22.35$  & $0.4942$ & $0.0069$ \\ \hline
$8$ & Broadcom Inc          & AVGO  & $481.73$  & $37.44$  & $26.24$  & $530$   & $37.70$  & $0.3933$ & $0.0032$ \\ \hline
$9$ & Oracle Corp           & ORCL  & $70.03$   & $36.06$  & $66.03$  & $78$    & $4.75$   & $0.3857$ & $0.0249$ \\ \hline
$10$ & Cisco Systems Inc     & CSCO  & $42.70$   & $33.10$  & $59.29$  & $47$    & $2.06$   & $0.3896$ & $0.0414$ \\ \hline
$11$ & Adobe Inc             & ADBE  & $371.94$  & $33.43$  & $52.08$  & $410$   & $39.25$  & $0.4569$ & $0.0038$ \\ \hline
$12$ & Salesforce Inc        & CRM   & $163.49$  & $34.66$  & $42.79$  & $180$   & $17.80$  & $0.4859$ & $0.0081$ \\ \hline
$13$ & Intel Corp            & INTC  & $37.21$   & $25.58$  & $51.98$  & $40$    & $3.40$   & $0.4041$ & $0.0393$ \\ \hline
$14$ & Qualcomm Inc          & QCOM  & $135.64$  & $100.34$ & $469.18$ & $150$   & $15.15$  & $0.4764$ & $0.0091$ \\ \hline
$15$ & Texas Instruments Inc & TXN   & $154.29$  & $16.71$  & $36.39$  & $170$   & $10.65$  & $0.4131$ & $0.0110$ \\ \hline
$16$ & Intuit Inc            & INTU  & $383.31$  & $27.55$  & $42.81$  & $420$   & $44.20$  & $0.4769$ & $0.0035$ \\ \hline
$17$ & AMD Inc               & AMD   & $77.52$   & $49.76$  & $120.59$ & $85$    & $13.10$  & $0.5152$ & $0.0130$ \\ \hline
$18$ & IBM Corp              & IBM   & $137.18$  & $8.84$   & $26.60$  & $150$   & $6.60$   & $0.3191$ & $0.0133$ \\ \hline
$19$ & Paypal Holdings Inc   & PYPL  & $71.36$   & $39.24$  & $47.09$  & $78$    & $11.80$  & $0.5119$ & $0.0144$ \\ \hline
$20$ & Netflix Inc           & NFLX  & $176.56$  & $72.39$  & $115.07$ & $195$   & $29.85$  & $0.5222$ & $0.0053$ \\ \hline
\end{tabular}
}
\caption{Key Financial Data as of $13$th July $2022$}
\label{tab:FinData}
\end{table}
\end{landscape}

\end{document}